\def\year{2021}\relax
\documentclass[letterpaper]{article} 
\usepackage{aaai21}  
\usepackage{times}  
\usepackage{helvet} 
\usepackage{courier}  
\usepackage[hyphens]{url}  
\usepackage{graphicx} 
\usepackage{natbib}  
\usepackage{caption} 
\usepackage{macros}
\usepackage{amsmath}
\usepackage{amsthm}
\usepackage{amsfonts}
\usepackage{mathtools}
\usepackage{xspace}
\usepackage{subcaption}

\usepackage{appendix}

\newtheorem{definition}{Definition}
\newtheorem{claim}{Claim}
\newtheorem{lemma}{Lemma}
\newtheorem{observation}{Observation}
\newtheorem{corollary}{Corollary}
\newtheorem{proposition}{Proposition}

\usepackage{xcolor}
\usepackage{soul}

\title{
Revisiting the Complexity Analysis of Conflict-Based Search: \\
New Computational Techniques and Improved Bounds
}
\author{
    Ofir Gordon,
    Yuval Filmus,
    Oren Salzman
    \\
}
\affiliations{
    Technion--Israel Institute of Technology, Computer Science Dept., Haifa, Israel \\
    ofirgo@cs.technion.ac.il, filmus.yuval@gmail.com, osalzman@cs.technion.ac.il
}

\begin{document}
\maketitle

\begin{abstract}
The problem of Multi-Agent Path Finding (MAPF) calls for finding a set of conflict-free paths for a fleet of agents operating in a given environment. Arguably, the state-of-the-art approach to computing optimal solutions is Conflict-Based Search (\cbs). In this work we revisit the complexity analysis of \cbs to provide tighter bounds on the algorithm's run-time in the worst-case. Our analysis paves the way to better pinpoint the parameters that govern (in the worst case) the algorithm's computational complexity.
Our analysis is based on two complementary approaches:
In the first approach we bound the run-time using the size of a Multi-valued Decision Diagram (\mdd)---a layered graph which compactly contains all possible single-agent paths between two given vertices for a specific path length.
In the second approach we express the running time by a novel recurrence relation which bounds the algorithm's complexity. We use generating functions-based analysis in order to tightly bound the recurrence.
Using these technique we provide several new upper-bounds on \cbs's complexity. The results allow us to improve the existing bound on the running time of \cbs for many cases. For example, on a set of common benchmarks we improve the upper-bound by a factor of at least $2^{10^{7}}$.
\end{abstract}

\section{Introduction} \label{sec:intro}
The \emph{Multi-Agent Path Finding problem} (\emph{MAPF})
is a well-studied problem, which attracts high interest among the robotics and AI community. It can be used to model many real-life applications, from automated warehouses \cite{warehouse_use}, through computer games \cite{Sturtevant2012} and to autonomous vehicles~\cite{vehicles_use}. Therefore, many efforts are invested in order to solve the problem as efficiently as possible, under different models and for different scenarios.

In the general version of MAPF \cite{Stern2019} we are given a graph~${G=(V,E)}$ with $n$ vertices and a set of $\agents$ agents~${A=\left\{a_1, a_2, ..., a_{\agents}\right\}}$. Each agent $a_i$ is provided with a start and a goal location, $(s_i,g_i)$ s.t.~${s_i,g_i \in V}$. Time is discretized and at every time-step an agent can either \emph{wait} in its current location or \emph{move} across an edge to an adjacent vertex. A \emph{feasible} solution is a paths set~${\mathcal{P}=\left\{p_1,p_2,...,p_{\agents}\right\}}$ such that $p_i$ is a path for agent $a_i$ from $s_i$ to $g_i$, and there is no conflict between any two paths in $\mathcal{P}$. We consider two types of conflicts---a \emph{vertex-conflict}, in which two agents occupy the same vertex at the same time-step, and an \emph{edge-conflict}, in which two agents traverse the same edge from opposite sides at the same time-step. An \emph{optimal} solution is a paths set~$\mathcal{P}$ which also optimizes some objective function. Arguably, the most common objective functions used for MAPF are:
\begin{enumerate}
    \item \emph{Makespan}--where we want to minimize the time in which the last agents arrives to its goal.
    \item \emph{Sum-of-Costs}--where we want to minimize the combined time it took for all agents to arrive at their goals.
\end{enumerate}

The task of finding an optimal solution is known to be NP-hard \cite{Yu2016} for both aforementioned objectives. The problem remains NP-hard even when~$G$ is a sub-graph of a planar grid graph \cite{Banfi2017}. Nevertheless, state-of-the-art optimal algorithms are able to effectively solve many non-trivial instances. Arguably, the most commonly-used algorithm for solving MAPF optimally is \emph{Conflict-Based Search} ($\cbs$) \cite{Sharon2015}. $\cbs$ first plans an initial (possibly infeasible) solution, and then systematically identifies and resolves conflicts. After $\cbs$ finds a conflict it applies a constraint that prohibits a conflicted agent from being in the location of the conflict at that certain time-step. Studies that followed, introduce different techniques that allow to empirically improve the algorithm's run-time \cite{Boyarski2015, Felner2018, Li2019heuristics, Zhang2020}.

Despite the ability of those techniques to cope with a wide range of non-trivial instances, there are many cases where $\cbs$ and its improvements cannot solve the problem even when allowed very long running times \cite{Kaduri2020}. Interestingly, many such empirically-hard instances do not exhibit notable differences from easy ones and identifying the exact source of (theoretical and empirical) hardness is an open question \cite{SalzmanS20}.

The original exposition of $\cbs$ \cite{Sharon2015} presented a (loose) upper-bound on the algorithm's complexity that is exponential both in the problem's parameters (number of agents $k$ and the number of vertices $n$ in the graph~$G$) and in the cost of an optimal solution. In this work we tighten this upper-bound and provide a new point-of-view on the analysis of a worst-case scenario for the algorithm. We believe that this is a first step towards improving our understanding on the problem's hardness which, in turn, will allow to design algorithms that can solve a wider range of instances.

We suggest two novel approaches to improve the algorithm's worst-case complexity analysis. 
In the first approach we improve the (existing) upper-bound on the number of possible constraints that \cbs might need to apply in order to find a solution. We do it by bounding the size of a \emph{Multi-valued Decision Diagram} (\mdd)---a layered graph which compactly contains all possible paths between two vertices for a specific path length \cite{Sharon2013}. 
In the second approach we express the run-time of \cbs using a novel recurrence relation which bounds the algorithm's complexity. We compute the \emph{generating function}~\cite{generatingFuncions} of the recurrence and use it to tightly bound its value in order to obtain a tighter bound on the complexity of \cbs.

Combining the results from both approaches provides us with new tighter upper-bounds for the worst-case complexity of \cbs. Beyond the new bounds, we anticipate that the computational tools we introduce will allow to obtain future improvements to the upper-bound. For example, this can be obtained via tighter bounds on the recurrence relation, or by improving the analysis of an \mdd size.

\section{Setting and Background}
Given an optimal solution $\mathcal{P}$, denote by $\mathcal{T}(p)$ the time that a single-agent's path $p \in \mathcal{P}$ terminates (note that wait moves are counted as a timestep in a path $p$). Now, set $\optcost$ to be the latest  time that a single-agent's path $p \in \mathcal{P}$ terminates. Namely,~${\optcost = \max_i \left\{\mathcal{T}(p_i)\right\}}$.

Under the minimal-makespan objective, $\optcost$ constitutes the cost of the optimal solution $\mathcal{P}$. Thus, for the rest of this paper we will consider $\optcost$ as the cost of an optimal solution of the problem, to be used in the complexity analysis.
Notice that~${\agents \optcost}$ constitutes an upper-bound on the cost of an optimal solution for the sum-of-costs objective. We further discuss the applicability of our results for the sum-of-costs objective in Sec.~\ref{sec:discussion}.

\subsection{CBS and its Complexity Analysis} \label{sec:background_cbs}
\emph{Conflict-Based Search} ($\cbs$) \cite{Sharon2015} is a two-level search algorithm which works as follows---first it finds an optimal path for each agent independently using some single-agent search algorithm like $\astar$ \cite{Hart1968}. $\cbs$ then works to resolve conflicts that occur in the solution: in the high-level search it preforms a best-first search upon a constructed conflicts-tree~($\ct$). Each node in the $\ct$ consists of a \emph{solution}, the solution's \emph{cost} and a set of \emph{constraints} imposed on the agents. A constraint is either a \emph{vertex-constraint} of the form $\negconst{a}{v}{t}$, which prohibits agent $a$ from being at vertex $v$ at time-step $t$, or an \emph{edge-constraint} of the form $\edgeconst{a}{u}{v}{t}$, which prohibits agent $a$ from crossing the edge $(u,v)$ at time-step $t$. We refer to such constraints as \emph{negative constraints}.

In each iteration, \cbs selects an unexpanded \ct node with a lowest cost. It then finds a conflict that occurs between two agents in the node's solution. It splits the \ct node into two child-nodes, each with a constraint on one of the agents that were involved in the conflict. It then runs the low-level search to construct a new solution in each child node, that does not violate the new constraint, by running a single-agent search algorithm like~$\astar$.

The basic version that we just presented was recently improved using \emph{positive constraints} \cite{Li2019Disjoint}. In a positive constraint $\posconst{a}{v}{t}$, agent $a$ is required to be at vertex $v$ at time-step $t$. When using positive constraint with \cbs, a \ct node is split into two child nodes using a positive and negative constraints forcing and forbidding the conflicted agent to be at a vertex or edge at a certain time-step, respectively.

An analysis of the worst-case time-complexity of \cbs was originally presented by \citet{Sharon2015}. They show that \cbs's complexity can be decomposed to bounding the size of the \ct and the complexity of the single-agent search in each low-level iteration. 
We refer to the size of the \ct in a worst-case scenario as the \emph{high-level search complexity}.
The low-level search complexity corresponds to running $\astar$ for a single agent. For the rest of this paper we focus on analyzing the high-level search complexity, thus, a complete upper-bound on \cbs's complexity can obtained by simply multiplying any of the following results with the complexity of a single agent's $\astar$-search.

The original analysis uses the assumption that each agent can potentially be in every vertex at every time-step. This bounds the number of (negative) constraints that \cbs might need to apply by $\bigo(n \agents \optcost)$. At each \ct node exactly one constraint is added. Thus, the number of possible constraints bounds the depth of the \ct, and gives an overall bound on \cbs's running time of~${\bigo(n \optcost \cdot 2^{\agents n \optcost})}$.%
\footnote{The original paper contains a minor error in the calculation of the upper bound. The bound presented here is the new bound whose validity was verified with one of the authors.
Similarly, the oversight regarding not accounting for edge constraints (explained shortly) was also discussed and verified with one of the authors in the original \cbs paper.} 
In the rest of this paper, we refer to this analysis and bound as the \emph{original analysis} and \emph{original upper-bound}, respectively.

It is important to note that the original analysis does not account for the possibility that \cbs would apply edge-constraints (in order to resolve conflicts).
It is possible that in a worst-case scenario the algorithm would require not only to prevent any agent from occupying any vertex in the graph at each time-step, but also from crossing each edge of the graph, in order to find the optimal solution. 
Therefore, accounting for edge constraint should further increase the theoretical upper-bound.
For clarity of exposition, we present our tools for analysing \cbs's complexity with the same assumption, i.e., that only vertex-constraints are considered.
Nevertheless, we address this issue in Sec.~\ref{sec:discussion} and show how to incorporate edge-constraints in the complexity analysis.

\subsection{Multi-valued Decision Diagram (MDD)}
The \emph{multi-valued decision diagram} $\mdd_i^C$ is a layered graph that consists of  $\optcost$ layers, which compactly contains all possible paths of agent $a_i$ of cost at most $\optcost$ from $s_i$ to $g_i$ \cite{Sharon2013}. 
A vertex~${v \in V}$ appears at the $t$'th layer of $\mdd_i^{\optcost}$ if it is reachable from $s_i$ and $g_i$ in $t$ and~${C-t}$ steps, respectively. 
Finally, the \emph{size} $\mddsize$ of an \mdd, represents the total number of \mdd nodes and the size $\mddsize_t$ of the $t$'th layer is the number of \mdd nodes in that layer.

$\mdd$ graphs are commonly-used for different purposes in MAPF algorithms, since they can be constructed efficiently for a given cost and their compact representation contains information that can help improve the identification and classification of conflicts \cite{Li2019heuristics, Zhang2020}. In this work we use $\mdd$s to bound the number of possible constraints that might need to be applied on a single agent during a $\cbs$ execution.

\subsection{Generating Functions for Bounding Recurrence Relations} \label{sec:gen-func-analysis}
Generating functions are a well-known mathematical tool which, among other things, can be used to bound recurrence relations. 
Formally, a generating function of a sequence~${a_0, a_1, a_2,\dots}$ with the general element denoted by~$a_r$, is the function~${\funcone{F}{x} = \sum_{r \geq 0} {a_r x^r}}$, i.e., the sequence elements are the coefficients of the series expansion of~$\funcone{F}{x}$. This notion can be extended for a sequence (or recurrence relation) with multiple variables. For instance, given a recursion~$\functwo{T}{r}{s}$ which defines a sequence, a possible generating function for it will be of the form~${\functwo{F}{x}{y} = \sum_{r,s \geq 0} {\functwo{T}{r}{s} x^r y^s}}$.
For further details on generating functions see, e.g., the book by~\citet{generatingFuncions}.

Given a generating function for a specific sequence, there can be many methods which allow to utilize the function in order to bound the value of the sequence at a certain index. 
These different methods are dependent on the sequence and the obtained function and there is no guarantee that a certain method could always be applied for this purpose.

\citet{PemantleW08} provide one approach for dealing with recursions of multiple variables which we briefly describe (additional details are presented mainly in Sec.~3 of the aforementioned paper) as it will be a key technique used to obtain our complexity bounds.
Assume that we are given a recursion~$\functwo{T}{r}{s}$ and a matching generating function for it~${\functwo{F}{x}{y} = \sum_{r,s \geq 0} {\functwo{T}{r}{s} x^r y^s}}$ that can be expressed by the following form: $\functwo{F}{x}{y} = \frac{\functwo{G}{x}{y}}{\functwo{H}{x}{y}}$. Denote by $H_z$ the partial derivative of $H$ for $z$ (where $z$ can be a sequence of $x$ and $y$).
The first step calls for finding \emph{critical points}, which are given by the solutions in the positive quadrant (i.e., $x,y \geq 0$) for the following system:
\begin{equation} \label{eq:critical-points-system}
    \begin{cases}
        H = 0 \\
        s x \partderv{H}{x} = r y \partderv{H}{y}.
    \end{cases}
\end{equation}
Denote the critical points by $q_1, q_2,\dots,q_m$. 
Each point~${q_i = \point{x_i}{y_i}}$ \emph{contributes} a certain factor to the approximation of $\functwo{T}{r}{s}$, and this contribution can be calculated according to the point's multiplicity. The exact way each point contributes to the bound is detailed by \citet{PemantleW08} and in \arxiv{the extended version of this paper~\cite{ExtendedVersion}}{Appendix.~\ref{appendix:gen-func-method}}.

Assume that the contribution of $q_i$ is given by 
$\functwo{T_i}{r}{s}$ for each $1 \leq i \leq m$, then the analysis suggests that the asymptotic growth of~$\functwo{T}{r}{s}$ can be tightly approximated by one of the factors which is given by the critical point's contribution.


\section{CBS's Complexity Analysis using MDDs} \label{sec:mdd}
Recall that the original analysis was obtained by bounding the number of possible (vertex) constraints that \cbs may apply.
In addition, as explained in Sec.~\ref{sec:background_cbs}, the original analysis did not account for edge constraints as a possible mean that can be used by the algorithm. 
We temporarily  limit our analysis to account for vertex constraints only and defer handling edge constraints to  Sec.~\ref{sec:discussion}.

We suggest a new approach to bound this number of possible constraints that \cbs may apply, using the following observation:

\begin{observation} \label{obs:const_mdd_bound}
    Given an agent $a_i$ and an optimal solution's cost $\optcost$, the maximal number of negative constraints that $\cbs$ may apply on $a_i$ is bounded by the size of $\mdd_i^{\optcost}$.
\end{observation}

Obs.~\ref{obs:const_mdd_bound} holds as \cbs may only apply a constraint on agent~$a_i$ at vertex $v$ for time-step $t$ if $a_i$ can reach $v$ from~$s_i$ within $t$ time-steps and still reach $g_i$ in $C-t$ time-steps, which is the exact definition of an \mdd node.

From Obs.~\ref{obs:const_mdd_bound} we obtain the following corollary:
\begin{corollary} \label{cor:mdd_bound}
    Let $\mddsize$ denote the maximal size of an agent's \mdd in a given instance. The size of $\cbs$'s conflict-tree is bounded by~${\bigo(2^{\agents \mddsize})}$ for any execution of the algorithm on this instance. This implies a similar bound on the algorithm's high-level search complexity.
\end{corollary}

Cor.~\ref{cor:mdd_bound} can be used to recover the original analysis of \cbs---a (loose) bound on $\mddsize$ can be obtained by bounding the size of any \mdd layer by $n$. Thus $\mddsize = \bigo(n \optcost)$ which gives us the original bound of~${\bigo(2^{\agents n \optcost})}$.  

We present two (tighter) bounds on $\mddsize$. The first (Sec~\ref{sec:mdd_size_worst_case}) removes the number of environment vertices from the complexity analysis, eliminating the possibility to deem a problem computationally hard just by adding inconsequential vertices to the environment.
See Fig.~\ref{fig:difficult_instance} for an example for such instance in which expressing the bound using the size of the graph might blow-up the result unnecessarily. In the presented instance, an optimal solution would require the agents to switch places in row ``F''. It can be observed that adding free nodes on the left, right or bottom of the grid won't affect the instance's difficulty.

The second bound accounts for the structure of $G$. In addition to the obtained bound, it demonstrates a complexity analysis restricted to a specific setting. This allows to (potentially) obtain tighter bounds on the size of an $\mdd$ which, in turn, provides tighter bounds on \cbs's complexity for a given setting of interest.

In the following sections, we assume that $G$ is a full~${\sqrt{n} \times \sqrt{n}}$ grid with no blocked vertices and~${s_i= g_i}$ for some agent $a_i$ (this serves as an upper bound on the size of~${\mdd_i^{\optcost}}$ for any other instance).

\begin{figure}[t]
    \centering
    \includegraphics[width=0.4\columnwidth]{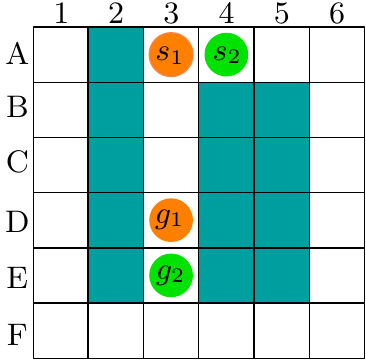}
    \caption{Empirically difficult instance for $\cbs$ with two agents. The colored squares represent blocked nodes.}
    \label{fig:difficult_instance}
\end{figure}

\subsection{Upper-Bound on the Size of an MDD} \label{sec:mdd_size_worst_case}
For the simplicity of the exposition, we restrict the discussion in this section to MAPF instances on infinite~{$\text{4-connected}$} grids \cite{Banfi2017, Stern2019}.
That is, we consider the setting where an agent can move in four directions from any vertex in the graph. 
Nonetheless, we emphasize that the technique we use to bound \cbs's complexity can be used to bound the size of an \mdd for any environment.

For any optimal path, an agent~$a_i$ can't be located at any vertex within distance larger than~${\optcost/2}$ from its start~$s_i$ or goal~$g_i$. This implies a symmetry on the structure of~${\mdd_i^{\optcost}}$---the last $\lfloor{\optcost/2}\rfloor$ layers form a mirror-image of the first $\lfloor{\optcost/2}\rfloor$ layers. The number of vertices on a grid which are reachable from $s_i$ within exactly $t$ steps is~$4t$ (see Fig.~\ref{fig:grid_distances}). At time-step $t$, $a_i$ can be located at any vertex within distance at most~$t$ from $s_i$. Therefore, we sum the number of reachable vertices in the range from one to~$t$ (excluding~$s_i$). For any $t \leq C/2$ the size of the $t$'th layer in~${\mdd_i^{\optcost}}$ is:

\begin{equation} \label{eq:mdd_layer}
    \mddsize_t \leq \sum_{i=1}^{t} {4i} = {2t(t+1)}.
\end{equation}
Given the aforementioned symmetry:
\begin{equation} \label{eq:mdd_size}
    \mddsize \leq 2 \cdot \sum_{t=1}^{{C/2}} {2t(t+1)} = \frac{C^3+6C^2+8C}{6} = \bigo(C^3).
\end{equation}

We assume for simplicity that $\optcost$ is even, if $\optcost$ is odd then the size of the middle layer ($\bigo(\optcost^2)$ according to Eq.~\ref{eq:mdd_layer}) needs to be added to the result of Eq.~\ref{eq:mdd_size}.

By placing the result from Eq.~\ref{eq:mdd_size} in Cor.~\ref{cor:mdd_bound} we obtain that:

\begin{claim} \label{claim:first_bound}
    The high-level search complexity of $\cbs$ on grid graphs is bounded by~${\bigo\left(2^{\agents \optcost^3}\right)}$.%
    \footnote{Note that in general it does not hold that $2^{\bigo(m)} = \bigo{\left(2^m\right)}$. The reason for which it does hold in this case is since the hidden constant in the Big-O notation in Eq.~\ref{eq:mdd_size} is smaller than 1.}
\end{claim}
The bound in Claim.~\ref{claim:first_bound} provides a tighter estimation for the algorithm's complexity for any instance where:~${\optcost^3 < n \cdot \optcost \Rightarrow \optcost < \sqrt{n}}$. 
In addition, this result removes $n$ from the bound expression.

\begin{figure}[t]
    \centering
    \includegraphics[width=0.4\columnwidth]{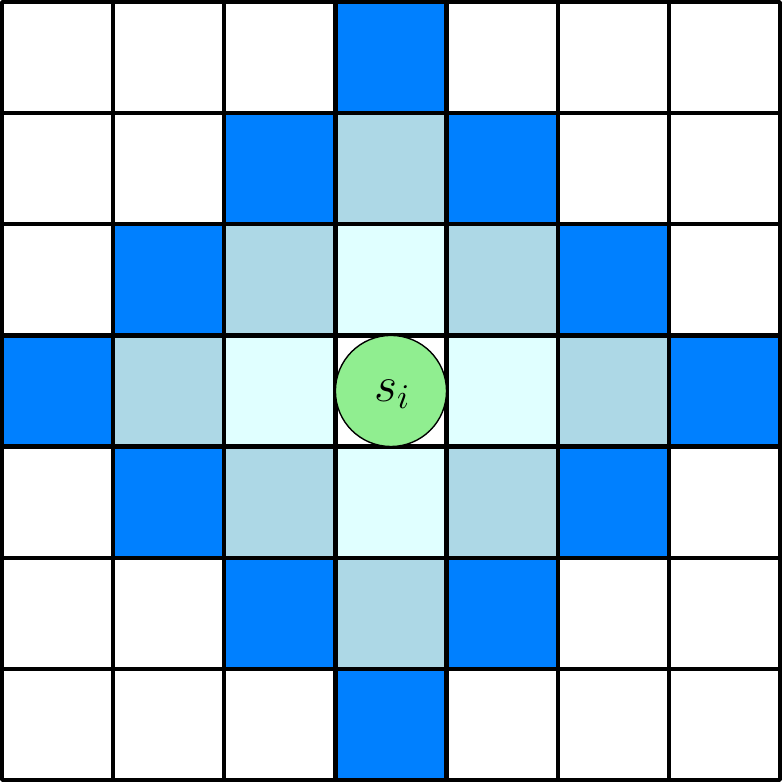}
    \caption{Illustration of vertices reachable for an agent $a_i$ located at $s_i$ within 1-3 time-steps, on a 4-connected grid.}
    \label{fig:grid_distances}
\end{figure}

\subsection{MDD Size Based on the Graph's Radius} \label{sec:mdd_based_radius}
\begin{definition} \label{def:radius}
The \textbf{distance} $dist(u,v)$ between two vertices $u,v$ in a graph is the number of edges on a shortest path between them. The \textbf{radius} of a graph is~${\radius = \min\limits_{u \in V}{\max\limits_{v \in V}{\{dist(u,v)\}}}}$. A vertex $u$ for which it holds that $\forall v \in V:\, dist(u,v) \leq \radius$ is called a \textbf{center} vertex.
\end{definition}

In a complete square grid of size~$n$, we have that~${\radius = \sqrt{n} - 1}$, with the center in the $\lceil{\frac{\sqrt{n}}{2}}\rceil$'th row and column for an odd value of $\sqrt{n}$. When $\sqrt{n}$ is even, there is no single center vertex. For simplicity, we assume that $\sqrt{n}$ is odd.

\begin{observation} \label{obs:radius}
 A layer of size $n$ exists in $\mdd_i^{\optcost}$ only if~${\optcost \geq 2 \radius}$ (note that a layer's size can't exceed $n$).
\end{observation}

Obs.~\ref{obs:radius} allows us to characterize settings for which it is possible to refine our previous analysis. For the first and last~$\radius$ layers of $\mdd_i^{\optcost}$ we bound a layer's size using Eq.~\ref{eq:mdd_layer}. Using~Obs. \ref{obs:radius}, the remaining layers are the only layers with size $n$. This gives us the following bound for $\mddsize$, for cases where~${\optcost=2\radius + \delta}$ for some $\delta \in \mathbb{N}$:
\begin{equation} \label{eq:mdd_radius}
\begin{split}
    \mddsize & \leq \delta n + 2 \cdot \sum_{t=1}^{\radius} {2t(t+1)} \\ &= \frac{4}{3} \cdot \radius (\radius+1) (\radius+2) + \delta n 
\end{split}
\end{equation}

The expression $\frac{4}{3} \cdot \radius (\radius+1) (\radius+2)$ is smaller than $2 \cdot \radius^3$ for any $\radius \geq 7$.
By placing the result from Eq.~\ref{eq:mdd_radius} in Cor.~\ref{cor:mdd_bound} we obtain that:
\begin{claim} \label{claim:radius}
    The high-level search complexity of $\cbs$ on grid graphs with radius $\radius \geq 7$ where~${\optcost=2\radius + \delta}$ for some $\delta \in \mathbb{N}$ is bounded by~${\bigo\left(2^{\agents \cdot (2\radius^3 + \delta n) }\right)}$.
\end{claim}

Claim.~\ref{claim:radius} not only allows to express the bound in terms of a new (and arguably, more relevant) parameter (namely, the radius of a graph), it also provides a slightly tighter bound on the overall complexity for full grid graphs, as long as the graph's radius is not too small. The new bound is tighter than the original bound for cases where:~${\agents \cdot (2\radius^3 + \delta n) < \agents n \cdot (2\radius +\delta) \Longrightarrow \radius < \sqrt{n}}$ (which, indeed holds for complete grids). The hidden constant in the Big-O notation in both the new and the original bounds is small and does not affect the asymptotic comparison between them.

\section{Complexity Analysis for CBS using a Recurrence Relation} \label{sec:rec-analysis}

We introduce a novel recurrence relation which bounds the high-level search complexity. More precisely, it bounds the maximal number of \ct nodes that might be generated during the high-level search. We provide an upper-bound on this recurrence relation that allows to improve the original bound on the run-time of \cbs for many cases.

Our improved bound incorporates the fact that recent \cbs variants use \emph{positive constraints} (Sec.~\ref{sec:background_cbs}). This is in contrast to the original analysis that only considers negative constraints.
However, the method in which the recursion is defined is not tied to positive constraints. We believe that tighter bounds may be obtained in the future by defining a similar recursion for alternative implementations of \cbs (such as \cbs with symmetry-breaking~\cite{Li2019symmetry}).

\subsection{Recurrence Relation which Bounds CBS's  Worst-Case Complexity}
Given a MAPF instance with~$\agents$ agents on a graph of size~$n$ where the optimal cost of a solution is $\optcost$, our goal is to bound the maximal number of \ct nodes that might be generated by \cbs after applying a given number of positive and negative constraints. \cbs will terminate if:
\begin{enumerate} 
    \item All possible negative constraints were applied (this assumption is similar to the one used for the original bound).
    \item The algorithm applied $\optcost$ positive constraints on each of the $\agents$ agents.
\end{enumerate}

Note that any (positive or negative) constraint applied to a \ct node cannot be applied to any of its children in the \ct. In addition, if agents $a_i$ and $a_j$ were found to be in a conflict at vertex $v$ at time-step $t$, applying a positive constraint on agent $a_i$ implies that the negative constraints $\negconst{a_i}{v}{t}$ and~$\negconst{a_j}{v}{t}$ cannot appear in the sub-tree of the \ct node.

From the above we get the following recurrence relation:
\begin{lemma} \label{lemma:recursion}
    Let $\recneg$ and $\recpos$ denote the maximal number of negative and positive constraints that \cbs may apply before it is bound to terminate, respectively. Then, the high-level complexity of \cbs is bounded by:
    \begin{equation} \label{eq:recurrence}
        T(\recneg,\recpos) \leq
        \begin{cases}
            1, & \recneg=0 \text{ or } \recpos=0 
            \\ 
            3, & \recneg=1 \text{ and } \recpos > 0 
            \\
            T(\recneg-1,\recpos) + \\\,\,\, T(\recneg-2,\recpos-1) + 1, & \text{ else}.
        \end{cases}
    \end{equation}
\end{lemma}

For $\recneg=0 \text{ or } \recpos=0$ we get that one of the aforementioned conditions for termination holds, therefore, this respected node is a leaf node.
For $\recneg = 1$ there is still a single negative constraint left to apply, so the node can be split only one more time, creating two additional leaves (and the node itself is also counted). 
For any other inner-node, the algorithm would split it according to a conflict by applying a negative constraint on one branch, and a positive constraint for the other branch.
Note that when applying a negative constraint (i.e., reducing $\recneg$ by one) it does not imply that a positive constraints has been applied (therefore, in the first component of the recurrence step $\recpos$ does not change).

We present two techniques for upper-bounding the recursion presented in Eq.~\ref{eq:recurrence}, which in turns provide an upper-bound for \cbs's complexity.

\subsection{Induction-Based Bound}
\begin{claim} \label{claim:recurrence_bound}
    For any $(\recneg$, $\recpos)$ s.t. $\recneg \geq 1$ and $\recpos \geq 1$ it holds that:
    \begin{equation} \label{eq:recurrence_bound}
        T(\recneg,\recpos) \leq 3 \cdot \recneg^\recpos.
    \end{equation}
    Which implies that $T(\recneg, \recpos) = \bigo(\recneg^\recpos)$.
\end{claim}

\begin{proof} [Proof sketch]
    The proof is by induction over pairs $(\recneg,\recpos)$, assuming an order where $(\recneg_1,\recpos_1) \geq (\recneg_2,\recpos_2)$ if ${\recneg_1 \geq \recneg_2}$ and~${\recpos_1 \geq \recpos_2}$ (there exists such an order on pairs where~${\recneg,\recpos \in \mathbb{N}}$).
    
    \noindent\underline{Base:}
    \begin{equation*}
    \begin{split}
            T(1,\recpos) &\leq 3 \leq 3\cdot 1^{\recpos}. \\
            T(\recneg,1) &\leq T(\recneg - 1,1) + T(\recneg - 2,0) + 1 =
            T(\recneg - 1,1) + 2 \\&
            \leq \dots \leq T(1,1) + 2 (\recneg - 1) = 2 \recneg + 1 \leq 3\cdot\recneg^1.
    \end{split}
    \end{equation*}
    
    \noindent\underline{Step:} We assume that the claim holds for all pairs smaller than $(\recneg,\recpos)$ and prove for $(\recneg,\recpos)$:
    \begin{equation*}
    \begin{split}
            T(\recneg,\recpos) &\leq T(\recneg-1,\recpos) + T(\recneg-2,\recpos-1) + 1 \\& 
            \mathrel{\mathop{\leq}\limits_{\text{i.h.}}}
            3(\recneg-1)^\recpos + 3(\recneg-2)^{\recpos-1} + 1 \\&
            \mathrel{\mathop{\leq}\limits_{*}}
            3(\recneg-1)^{\recpos} + 3(\recneg-1)^{\recpos - 1} \\ &= 
            3[(\recneg-1)^{\recpos-1} \cdot (\recneg-1) + (\recneg-1)^{\recpos - 1}]\\ &= 
            3(\recneg-1)^{\recpos-1} \cdot (\recneg-1+1) \\ & =
            3\recneg \cdot (\recneg-1)^{\recpos-1} \leq 3\recneg \cdot \recneg^{\recpos-1} = 3\recneg^\recpos.
    \end{split}
    \end{equation*}
    Where~$*$ holds because $(\recneg-1)^{\recpos-1} \geq 1$ for $\recneg > 1$ and~${\recpos > 1}$.
\end{proof}

Recall that negative and positive constraints are bounded by~${\recneg = \agents \mddsize}$ and~${\recpos = \agents \optcost}$, respectively (where $\mddsize$ is the size of an $\mdd$ graph of a single agent). Placing those values in Eq.~\ref{eq:recurrence_bound} gives the following result:
\begin{equation} \label{eq:upper_bound}
    T(\agents \mddsize, \agents \optcost) \leq \bigo{\left((\agents \mddsize)^{\agents \optcost}\right)}.
\end{equation}

From Eq.~\ref{eq:upper_bound} we obtain the following lemma:
\begin{lemma} \label{lemma:cbs_bound_from_recurrence}
    The time-complexity of the high-level search of \cbs is bounded by~${\bigo{\left((\agents \mddsize)^{\agents \optcost}\right)}}$.
\end{lemma}

By taking~${\mddsize = n \optcost}$ (i.e., the bound on an \mdd size considered in the original analysis), we get an
upper-bound of~${\bigo{\left((\agents n \optcost)^{\agents \optcost}\right)}}$ in contrast to the original bound of~${\bigo(2^{n \agents \optcost})}$. Here again, the hidden constant in the Big-O notation in both bounds is small, thus, we omit it in the upcoming comparison between them (see end of Sec.~\ref{sec:bound-with-gf}).

\begin{table}[t]
\centering
\resizebox{1.0\columnwidth}{!}{
\begin{tabular}{|c | c c c | c c c |c|}
    \hline
    \rule{0pt}{3ex} Benchmark Category & $n$ & $\agents$ & $\optcost$ & ORG & REC+IND & REC+GF & $\frac{\text{ORG}}{\text{REC+GF}}$  \\
    \hline\hline
     \rule{0pt}{2.5ex} Warehouse & 9,776 & 8 & 120 & 
     $2^{10^{7}}$ & $2^{10^{5}}$ & $2^{10^{5}}$ & $\mathbf{2^{10^{7}}}$ \\
    \hline
     \rule{0pt}{2.5ex} Warehouse & 9,776 & 64 & 140 & 
     $2^{10^{8}}$ & $2^{10^{6}}$ & $2^{10^{5}}$ & $\mathbf{2^{10^{8}}}$ \\
     \hline
     \rule{0pt}{2.5ex} Warehouse & 38,756 & 128 & 250 & $2^{10^{10}}$ & $2^{10^{7}}$ & $2^{10^{5}}$ & $\mathbf{2^{10^{10}}}$ \\
     \hline
     \rule{0pt}{2.5ex} Warehouse & 38,756 & 256 & 250 & 
     $2^{10^{10}}$ & $2^{10^{7}}$ & $2^{10^{5}}$ & $\mathbf{2^{10^{10}}}$ \\
     \hline
     Room & 206,642 & 8 & 400 & 
     $2^{10^{9}}$ & $2^{10^{6}}$ & $2^{10^{6}}$ & $\mathbf{2^{10^{9}}}$\\
     \hline
     \rule{0pt}{2.5ex} Room & 206,642 & 8 & 500 & 
     $2^{10^{9}}$ & $2^{10^{6}}$ & $2^{10^{6}}$ & $\mathbf{2^{10^{9}}}$ \\
     \hline
     \rule{0pt}{2.5ex} Empty & 2,304 & 64 & 70 & 
     $2^{10^{8}}$ & $2^{10^{6}}$ & $2^{10^{4}}$ & $\mathbf{2^{10^{8}}}$ \\
     \hline
     \rule{0pt}{2.5ex} Empty & 2,304 & 128 & 80 & 
     $2^{10^{8}}$ & $2^{10^{7}}$ & $2^{10^{4}}$ & $\mathbf{2^{10^{8}}}$ \\
     \hline
     \rule{0pt}{2.5ex} Random & 3,687 & 64 & 100 & 
     $2^{10^{8}}$ & $2^{10^{6}}$ & $2^{10^{4}}$ & $\mathbf{2^{10^{8}}}$ \\
     \hline
     \rule{0pt}{2.0ex} Random & 3,687 & 128 & 100 & 
     $2^{10^{8}}$ & $2^{10^{7}}$ & $2^{10^{4}}$ & $\mathbf{2^{10^{8}}}$ \\
     \hline
\end{tabular}
}
\caption{A comparison between the different upper-bounds obtained using the original analysis (ORG), Lemma.~\ref{lemma:cbs_bound_from_recurrence} (REC+IND) and Prop.~\ref{prop:cbs-recursion-est-linear} (REC+GF), on standard benchmarks \cite{Sturtevant2012, Stern2019}.
The last column presents a lower bound on the ratio between our improved bound and the original bound, which reflects the improvement.
All bounds are calculated considering that $\mddsize = n \optcost$. Note that all actual bounds include a small constant multiplication factor, but the comparison in this table accounts only for the asymptotic factors.
}
\label{table:benchmarks}
\end{table}

\subsection{Generating Functions-Based Bound} \label{sec:bound-with-gf}
In our second approach, we present an alternative approach to bounding the recursion (Eq.~\ref{eq:recurrence}) using \emph{generating functions}. This, in turn, will allow us to obtain a tighter bound on \cbs's complexity.
Due to lack of space we only outline the analysis and refer the reader to \arxiv{the extended version of this paper~\cite{ExtendedVersion}}{Appendix.~\ref{appendix:gen-func-method}}. 

We start by introducing the generating function for~$\functwo{T}{\recneg}{\recpos}$. We then continue to follow the steps outlined by \citet{PemantleW08} to obtain a bound on~$\functwo{T}{\recneg}{\recpos}$.
We stress that the suggested analysis method is not applicable for formally proving the bound's correctness, but we can use it to deduce an asymptotic upper-bound which we then support empirically for a variety of different values.

The method by \citet{PemantleW08}, as described in Sec.~\ref{sec:gen-func-analysis}, calls for finding the contribution for the bound obtained by each critical point given from the solutions for Eq.~\ref{eq:critical-points-system}.
We first find the the contribution for each critical point. We then apply the following key observation which allows us to deduce a suggested upper-bound for \cbs. The observation follows from our analysis of the size of an \mdd (Sec.~\ref{sec:mdd_size_worst_case}):
\begin{observation} \label{obs:linear-dependency}
    For any MAPF instance, there is a linear dependency between $\recneg$, the maximal number of negative constraints and $\recpos$, the maximal number of positive constraints that \cbs can apply.
    Specifically,~${\recneg = n \cdot \left( \agents \optcost \right) = n \cdot \recpos}$.
\end{observation}
By applying Obs.~\ref{obs:linear-dependency} we get that one of the contribution factors obtained from the analysis gives a tight upper-bound on~$\functwo{T}{\recneg}{\recpos}$.

As explained, We start with presenting the generating function for this recursion, which is:
\begin{equation}
    \functwo{F}{x}{y}=\frac{1-x+2xy-x^2y}{(1-x)(1-y)(1-x-x^{2}y)}.
\end{equation}
We denote $F = G/H$ where
\begin{equation*}
    \begin{split}
        &G(x,y) = 1-x+2xy-x^2y, \\ 
        &H(x,y) = (1-x)(1-y)(1-x-x^{2}y),
    \end{split}
\end{equation*}
%
and solve Eq.~\ref{eq:critical-points-system} to find the critical points.
In this setting there are three such points:
\begin{equation*} \label{eq:critical-points}
    \begin{split}
        &q_1 \coloneqq \point{x_1}{y_1} = \left(\frac{-1+\sqrt{5}}{2}, 1\right), \\
        &q_2 \coloneqq \point{x_2}{y_2} = \left(1, 1\right), \\
        &q_3 \coloneqq \point{x_3}{y_3} = \left(\frac{\recneg-2\recpos}{\recneg-\recpos}, \frac{\recpos(\recneg-\recpos)}{(\recneg-2\recpos)^2}\right).
    \end{split}
\end{equation*}
We denote the matching contribution factor of each point~$q_i$ by $\functwo{T_i}{\recneg}{\recpos}$. 
%
Computing the exact contribution for each point is done according to \citet{PemantleW08}. 
This involves  basic (yet daunting) algebraic manipulations and is summarized in Lemma.~\ref{lemma:cbs-recursion-est-general}
(proof omitted).
\begin{lemma} \label{lemma:cbs-recursion-est-general}
    \begin{equation*}
        \begin{split}
            &\functwo{T_1}{\recneg}{\recpos} = 1, \\&
            \functwo{T_2}{\recneg}{\recpos} = \bigo(1) \cdot \left( \frac{1+\sqrt{5}}{2} \right)^{r}, \\&
            \functwo{T_3}{\recneg}{\recpos} = \frac{(\recneg-\recpos)^{\recneg-\recpos}}{(\recneg-2\recpos)^{\recneg-2\recpos}\cdot \recpos^{\recpos}} 
        \cdot
        \frac{2\recpos}{\recneg-2\recpos} 
        \cdot
        \sqrt{\frac{\alpha}{2\pi}},
        \end{split}
    \end{equation*}
    
    where $\alpha = \bigo{\left(\frac{\recneg^2}{s}\right)}$.
\end{lemma}

\begin{figure*}[t]
    \centering
    \begin{subfigure}{0.69\columnwidth}
        \includegraphics[width=\textwidth]{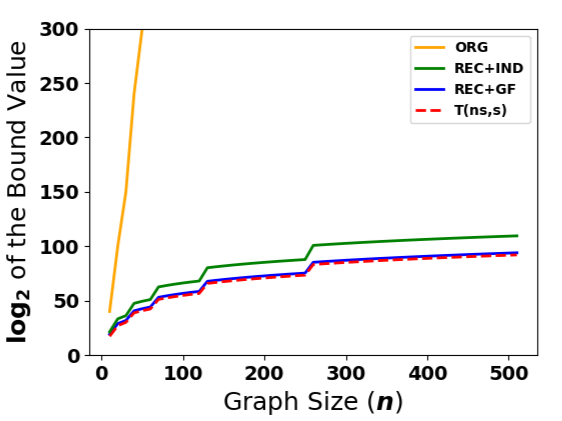}
        \caption{$s = \log_2{n}$}
        \label{fig:bounds-s-log-n}
    \end{subfigure}
    \begin{subfigure}{0.69\columnwidth}
        \includegraphics[width=\textwidth]{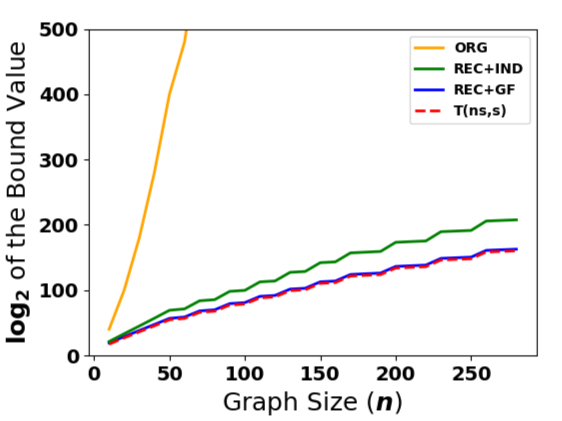}
        \caption{$s = \sqrt{n}$}
        \label{fig:bounds-s-sqrt-n}
    \end{subfigure}
    \begin{subfigure}{0.69\columnwidth}
        \includegraphics[width=\textwidth]{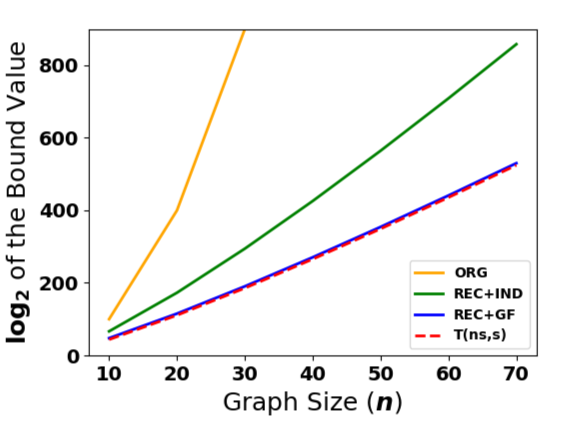}
        \caption{$s = n$}
        \label{fig:bounds-s-n}
    \end{subfigure}
    \caption{The $\log_2$ of the different bounds as a function of the graph size for different ratios between the graph's size ($n$) and the instance properties ($s=\agents \optcost$). 
    The two new bounds (REC+IND, REC+GF) are significantly lower than the original bound. 
    Notice that the approximation obtained from the generating functions analysis (REC+GF) indeed tightly bounds the recurrence~$T(ns,s)$.}
    \label{fig:bounds-comparison}
\end{figure*}

Following \citet{PemantleW08}, we can use Lemma.~\ref{lemma:cbs-recursion-est-general} to estimate the asymptotic growth of Eq.~\ref{eq:recurrence}. Specifically, this growth is likely to be estimated by one of the three terms.
Yet, using it to deduce an upper-bound for \cbs's complexity is not straightforward.

Fortunately, by applying Obs.~\ref{obs:linear-dependency} we can obtain an estimated bound on Eq.~\ref{eq:recurrence} that is tighter than the one obtained using the induction-based analysis~${( \text{Lemma.}}$~\ref{lemma:cbs_bound_from_recurrence}).
We do it by restricting the recurrence to values of $\recneg$ and $\recpos$ that can be attained in our MAPF setting. Specifically, using $\recneg=n\cdot \recpos$ in Lemma.~\ref{lemma:cbs-recursion-est-general} we have that,
\begin{proposition}
\label{prop:cbs-recursion-est-linear}
The high-level search complexity of \cbs for instances with $n \geq 4$ vertices, $\agents$ agents and an optimal solution cost $\optcost$ is bounded by $\bigo{\left((e n)^{\agents \optcost} \right)}$.
\end{proposition}

We approximate the value of $\functwo{T}{n\recpos}{\recpos}$ according to Lemma.~\ref{lemma:cbs-recursion-est-general} and have that
\begin{equation} \label{eq:cbs-recursion-est-linear}
    \begin{split}
        &\functwo{T_1}{n\recpos}{\recpos} = 1, \\&
        \functwo{T_2}{n\recpos}{\recpos} = \bigo(1) \cdot \left( \frac{1+\sqrt{5}}{2} \right)^{n\recpos}, \\& 
        \functwo{T_3}{n\recpos}{\recpos} = \left(\frac{(n-1)^{n-1}}{(n-2)^{n-2}} \right)^\recpos 
        \cdot 
        \frac{2}{n-2}
        \cdot
        \sqrt{\frac{\beta}{2\pi \recpos}},
    \end{split}
\end{equation}

\noindent
where $\beta = \bigo{\left( n^2 \right)}$.

The contribution to $\functwo{T}{n\recpos}{\recpos}$ from $q_1$ and $q_2$ (given by~${\functwo{T_1}{n\recpos}{\recpos} + \functwo{T_2}{n\recpos}{\recpos}}$)
is identical to the contribution from $q_3$ (given by $\functwo{T_3}{n\recpos}{\recpos}$) at $n_0 = \frac{\sqrt{5} + 2}{2} \approx 3.618033$.
In Sec.~\ref{sec:bounds-comparison} we empirically demonstrate that $\functwo{T_3}{n\recpos}{\recpos}$ indeed constitutes a tight bound for any $n > n_0$.

Therefore, we continue with the simplification of the expression given by $\functwo{T_3}{n\recpos}{\recpos}$. Since $\frac{2}{n-2} \cdot \sqrt{\frac{\beta}{2\pi}} = \bigo(1)$ and~${\left( \frac{(n-1)^{n-1}}{(n-2)^{n-2}} \right) < e n}$, we get the following result:
\begin{equation*}
    \functwo{T}{n\recpos}{\recpos} = \bigo{\left( (e n)^{\recpos} \right)},
\end{equation*}
\noindent
with a small hidden constant factor in the Big-O notation. 
By placing $\recpos = \agents \optcost$, which is the maximal number of positive constraints that needs to be applied by \cbs in the worst case, we get the desired bound.

Prop.~\ref{prop:cbs-recursion-est-linear} improves the original known bound of~${\bigo{\left(2^{n \agents \optcost}\right)}}$ for any set of values~${n, \agents}$ and $\optcost$.
Moreover, it is also tighter than the already-improved bound presented in Lemma.~\ref{lemma:cbs_bound_from_recurrence}. 
Notice that it allows to replace the asymptotic factor of $\agents \optcost$ in the base of the exponent with a constant~($e$), while also still eliminating the exponential dependency in $n$.

New bounds can also be obtained by combining the results from Sec.~\ref{sec:mdd} that bound the size of an \mdd.
For example, using Eq.~\ref{eq:mdd_size} we observe that there is a quadratic dependency between $\recneg = \agents \optcost^3$ and $\recpos = \agents \optcost$ on 4-connected grids.
By simply substituting $n$ with $\optcost^2$ in the bound obtained by Prop.~\ref{prop:cbs-recursion-est-linear} we have that,
\begin{corollary}
    The high-level search complexity of \cbs on 4-connected grids is bounded by $\bigo{\left((e \optcost)^{2 \agents \optcost} \right)}$.
\end{corollary}

\subsection{Empirical Comparison of Bounds}
\label{sec:bounds-comparison}
To demonstrate the difference between the new bounds and the original one, we evaluated their ratio for commonly-used benchmarks.
The results are summarized in Table~\ref{table:benchmarks}, which presents the worst-case complexity of \cbs's high-level search, for the original bound (ORG), the looser induction-based bound presented in Lemma.~\ref{lemma:cbs_bound_from_recurrence} (REG+IND) and the tighter generating functions-based bound presented in Prop.~\ref{prop:cbs-recursion-est-linear} (REC+GF).
It also presents the ratio between the original bound and the best newly-obtained bound, to demonstrate the improvement.

For all benchmarks, the new bound improves the original one by a factor of at least~${2^{10^{7}}}$.
More precisely, for any instance we examined, we obtain a significantly-tighter bound on the high-level complexity of \cbs.

To support the correctness of the bound obtained in Prop.~\ref{prop:cbs-recursion-est-linear}, we evaluated and compared its value to the recurrence's value for large values of $n$ and $\recpos$. We also use this evaluation to further demonstrate the difference between the new bounds and the original one.

Fig.~\ref{fig:bounds-comparison} presents a comparison (on a logarithmic scale) of the three bounds (ORG, REG+IND and REC+GF) and~$T(\recneg, \recpos)$ for different graph sizes $n$ for the setting $\recneg = n \recpos$ (see Obs.~\ref{obs:linear-dependency}). 
The three different figures reflect different settings of MAPF, where the dependency between the graph's size and the instance properties (number of agents and optimal solution cost) changes. 
Indeed for all different cases the same trend can be observed where the gap between the presented functions demonstrate the magnitude of the improvement obtained using our analysis. 
In addition, the figures serve as an empirical validation of the generating functions-based analysis---the asymptotic bound (REC+GF) tightly approximates the recurrence relation (Eq.~\ref{eq:recurrence}).

\section{Summary and Discussion}
\label{sec:discussion}
\subsection{Summary}
We presented two novel approaches for analyzing the worst-case complexity of \cbs's high-level search. 
In the first approach, by analysing the size of an agent's \mdd graph, we provided two new upper-bounds for \cbs's high-level search of $\bigo(2^{\agents \optcost^3})$ and $\bigo(2^{\agents \cdot (2\radius^3 + \delta n)})$. 
The latter is a bound for a setting where a solution's optimal cost is dependent on the radius $\radius$ of $G$. Our approach allows to seamlessly obtain tighter bounds on CBS's complexity given tighter bounds on~$\mddsize$. Such bounds may be obtained either by (i) better analyzing the general structure of an \mdd or (ii) restricting the analysis for a specific instance of interest.

In the second approach we presented the recurrence relation $T(\recneg, \recpos)$. An upper-bound on $T(\recneg, \recpos)$ constitutes a bound on the size of the \ct of \cbs, therefore bounding  the complexity of \cbs's high-level search. 
Using this approach we obtain a new general bound of~${\bigo((\agents \mddsize)^{\agents \optcost})}$.
When using $\mddsize = n \optcost$, we obtain the new induction-based bound of~${\bigo((\agents n \optcost)^{\agents \optcost})}$.

Using a generating functions-based bound, we obtain a tighter bound on the recurrence which, in turn, provides a tighter bound for \cbs. 
Observing that there exists a linear dependency between the number of negative and positive constraints, allows us to achieve further improvement, and eventually obtain a bound of~$\bigo{\left((e n)^{\agents \optcost} \right)}$, which improves the original bound on the algorithm by a significant factor for a wide range of standard benchmarks.

We believe that the recurrence relation can be further improved, in order to better express the real conditions of a worst-case scenario. An immediate step would be to try and account for tighter dependencies between the number of constraints that are being eliminated once a single constraint is applied.
In addition, revisiting the conditions on the recursion's parameters may allow to tighten the upper-bound on the recursion, and in turn, on the complexity of \cbs.

It is important to note that the new bounds are still somewhat loose and present a worst-case analysis.
However, our analysis paves the way to better pinpoint the parameters that govern (in the worst case) the algorithm’s computational complexity as well as analyze the complexity when restricted to certain settings.
Moreover it provides a general methodology that can be used to analyze different variants of the MAPF problem. 
For example, in the next sections we show how to seamlessly account for edge constraints (Sec.~\ref{subsec:edge-constraints}) as well as for settings that optimize the Sum-of-Costs objective (Sec.~\ref{subsec:soc}).


\subsection{\cbs with Edge-Constraints}
\label{subsec:edge-constraints}
Recall that the analysis we performed in Sec.~\ref{sec:mdd} and~\ref{sec:rec-analysis} (as well as the original analysis) accounted for vertex constraints only.
We now show simple approaches to account for edge constraints using the existing analysis.
Accounting for edge constraints directly is left for future work.

\subsection*{Counting edge constraints as vertex constraints}
Notice that an edge constraint implicitly defines two vertex constraints (forcing agent $a_i$ to traverse $(u,v)$ at time~$t$ corresponds to the constraint that $a_i$ has to be in vertex $u$ and~$v$ at times $t$ and $t+1$, respectively).
Thus, we can simply increase the number of vertex constraints by the number of possible edge constraints (which is twice the number of edges as each edge can be traversed in both directions).
Specifically, on 2D-grids, each node has at most $4$ outgoing and incoming edges, therefore, using the original analysis the number of negative constraints would increase to~${n \agents \optcost +8n \agents \optcost = 9n \agents \optcost}$ in the worst-case. 
For the general case where~${|E| \leq n^2}$, the number of negative constraints would increase to~${(2n^2 + n)\cdot\agents \optcost}$ in the worst-case. 
We emphasize that while the increase may seem negligible, the actual worst-case complexity is exponential in the number of negative constraints and the additional constraints would appear in the exponent of the original bound.
For example, the original bound of $\bigo(2^{nk\optcost})$ would be $\bigo(2^{9 \cdot nk \optcost})$.

A similar approach can be applied to the recursion-based bounds presented in Sec.~\ref{sec:rec-analysis}.
We can consider a total number of negative constraints of $\recneg = 9 n \agents \optcost$, which also includes the negative edge-constraint.
One can show that Eq.~\ref{eq:recurrence} still upper-bounds the number of expanded nodes in \cbs's high-level search.
Thus, by placing $\recneg = 9 n \agents \optcost$ in Prop.~\ref{prop:cbs-recursion-est-linear} we obtain a bound of $\bigo{\left(\left(9en\right)^{\agents \optcost}\right)}$ (the value $n_0$ from which this approximation to the recursion's value holds slightly changes as well).

It is important to note that accounting for edge constraints increases the relative improvement of the bounds we present over the original bound. This is because in our bound the additional work is reflected by increasing the base of the exponent and not the exponent itself as in the original bound.
\subsection*{Counting  \mdd edges}
In Cor.~\ref{cor:mdd_bound} the size of an \mdd is defined as the number of \mdd vertices. However, if we wish to account for edge constraints, the same result holds simply by changing the definition of the size of an \mdd  to be the number of \mdd vertices \emph{and} edges.
For instance, on a 4-connected grid, the maximal (outgoing) degree of each  \mdd node is five.
Therefore, the total number of \mdd edges is bounded by five times the number of \mdd vertices.
Using Eq.~\ref{eq:mdd_size}, the total number of vertex and edge constraints is bounded by $(1+5) \cdot \bigo\left(\optcost^3\right)$. 
We then update Claim.~\ref{claim:first_bound} to incorporate edge constraints, and get a bound of $2^{\bigo\left(\agents \optcost^3\right)}$ (with a small hidden constant factor slightly larger than 1).

\subsection{Analysis for the Sum-of-Costs (SoC) Objective}
\label{subsec:soc}
The bounds we provided (Sec.~\ref{sec:mdd} and~\ref{sec:rec-analysis}) were obtained and expressed using $\optcost$---an upper-bound on a single agent's path length in an optimal solution. 
In the setting where we seek to minimize the makespan, $\optcost$ is indeed an optimal solution's cost. 
Unfortunately, this is not the case for the SoC objective.

One way to use our results when using SoC as our optimization objective is to observe that $\agents \optcost$ is an upper-bound for the optimal solution's cost.
Namely, denote $\optcost' = k\optcost$ and use $\optcost'$ instead of $\agents \optcost$ in all the bounds presented.
For example, our generating function-based analysis for the SoC objectives yields the bound of~$\bigo{\left((e n)^{\optcost'} \right)}$. 
This is a slightly looser bound  expressed with an upper-bound on a single-agent's solution's optimal cost.

\subsection{Future Improvements for \cbs}
Throughout this paper we presented several observations, which improve our understanding regarding the hardness of the MAPF problem. For instance, in Sec.~\ref{sec:mdd} we discuss the dependency of the graph’s size on the problem’s complexity, and use it later to provide refined upper-bounds. We provide a new understanding that the worst-case complexity of \cbs depends more on the graph's radius rather on the graph's size.
Another such observation is reflected in the incorporation of positive constraints in our recurrence-based analysis. It focuses on the importance of positive constraints by demonstrating the huge reduction in the search tree’s size. 

We hope that these observations would allow, in addition to improving the theoretical analysis of \cbs, to be used to improve the algorithm in practice.
One possible way to approach this task is to use these observations to develop heuristics. For instance, by favoring expansion of \ct nodes with a large number of positive constraints or with a minimal amount of approximated work that is remained to be done by the algorithm.

\section*{Acknowledgments}
We wish to thank Roni Stern and the anonymous reviewers of earlier versions of this paper for helpful discussions and comments that contributed to this paper.

\newpage

\begin{appendices}
\section{Generating Functions Based Analysis for Recurrence Asymptotic Approximation} \label{appendix:gen-func-method}
\vspace{1.0mm}
\subsection*{The Recurrence Relation Generating Function}
Eq.~\ref{eq:recurrence} presents a recurrence relation $\functwo{T}{\recneg}{\recpos}$ which forms an upper-bound on the complexity of the high-level search of the \cbs algorithm.
In order to get an estimated upper-bound on $T$, we solve the recurrence corresponding to the case in which all inequalities in Eq.~\ref{eq:recurrence} are tight.

The first step in our analysis requires finding the generating function of  $\functwo{T}{\recneg}{\recpos}$. 
This means that $\functwo{T}{\recneg}{\recpos}$ is the coefficient of $x^{\recneg} y^{\recpos}$ in $F$: 
\begin{equation} \label{eq:gen-func-general}
    \functwo{F}{x}{y} = \sum_{\recneg,\recpos \geq 0}{\functwo{T}{\recneg}{\recpos}x^{\recneg} y^{\recpos}}.
\end{equation}
Obtaining the exact form of $\functwo{F}{x}{y}$  is done using the following equation (for further reading about finding generating function for recurrence with multiple variable, we refer the readers to Chapter~{1.5} in \citet{generatingFuncions}):
\begin{equation*}
    \begin{split}
        \sum_{\recneg,\recpos \geq 0}{\functwo{T}{\recneg+2}{\recpos+1}x^{\recneg} y^{\recpos}} &=
        \sum_{\recneg,\recpos \geq 0}{\functwo{T}{\recneg+1}{\recpos+1}x^{\recneg} y^{\recpos}} \\
        &+\sum_{\recneg,\recpos \geq 0}{\functwo{T}{\recneg}{\recpos}x^{\recneg} y^{\recpos}} \\
        &+\sum_{\recneg,\recpos \geq 0}{x^{\recneg} y^{\recpos}}.
    \end{split}
\end{equation*}
The computation involves algebraic simplification and subtitution of the sums with $\functwo{F}{x}{y}$ according to Eq.~\ref{eq:gen-func-general}:
\begin{equation*}
    \begin{split}
        &\frac{1}{x^2y} \left(\functwo{F}{x}{y} - 1 -\frac{y}{1-y} - \frac{x}{1-x} - \frac{3xy}{1-y} \right) =  \\
        &\,\,\,\,\,\, \functwo{F}{x}{y} +\frac{1}{xy}\cdot \left(\functwo{F}{x}{y} - 1 - \frac{y}{1-y} - \frac{x}{1-x}\right) + \\ &\,\,\,\,\,\,\,\,\frac{1}{\left(1-x\right)\left(1-y\right)},
    \end{split}
\end{equation*}
which, eventually gives the following generating function for the recurrence:
\begin{equation} \label{eq:generating-function}
    \functwo{F}{x}{y}=\frac{1-x+2xy-x^2y}{(1-x)(1-y)(1-x-x^{2}y)}.
\end{equation}

\vspace{1.0mm}
\subsection*{Recurrence Approximation Analysis}
\subsubsection*{Critical Points}
In order to get a closed-form formula which approximates the value of $\functwo{T}{\recneg}{\recpos}$ we follow the analysis by \citet{PemantleW08}.
As stated in Sec.~\ref{sec:bound-with-gf}, this method does not allow to prove the bound's correctness for a recurrence relation of the form that we have, but it does allow us to obtain an estimated bound, which we empirically evaluate to be percise.

The first step is to express the function as the ratio of~${F = G/H}$, that is:
\begin{equation*}
    \begin{split}
        &\functwo{G}{x}{y} = 1-x+2xy-x^2y, \\ &\functwo{H}{x}{y} = (1-x)(1-y)(1-x-x^{2}y).
    \end{split}
\end{equation*}
We use $\partderv{H}{x},\partderv{H}{y},\partderv{H}{xx},\partderv{H}{yy},\partderv{H}{xy}$ for the partial derivatives of~$H$ with respect to the sub-scripted variables:
\begin{equation} \label{eq:partial-derivatives}
    \begin{split}
        &\partderv{H}{x} = (1-y)(3x^2y-2x(y-1)-2), \\
        &\partderv{H}{y} = (1-x)(x^2(2y-1)+x-1), \\
        &\partderv{H}{xx} = -2(y-1)((3x-1)y+1), \\
        &\partderv{H}{yy} = -2x^2 (x-1), \\
        &\partderv{H}{xy} = x^2(3-6y)+4x(y-1)+2.
    \end{split}
\end{equation}

We need to find the critical points which are given by the solutions for the following system in the positive quadrant~{($x,y>0$)}:
\begin{equation*} \label{eq:crititcal-system}
    \begin{cases}
        H = 0 \\
        s x \partderv{H}{x} = r y \partderv{H}{y}.
    \end{cases}
\end{equation*}
There are three solutions to the system in the positive quadrant:
\begin{equation} \label{eq:appx-critical-points}
    \begin{split}
        &(x_1, y_1) = \left(\frac{-1+\sqrt{5}}{2}, 1\right) \\
        &(x_2, y_2) = \left(1, 1\right) \\
        &(x_3, y_3) = \left(\frac{\recneg-2\recpos}{\recneg-\recpos}, \frac{\recpos(\recneg-\recpos)}{(\recneg-2\recpos)^2}\right).
    \end{split}
\end{equation}

The third solution is missing when $\recneg = \recpos$ or $\recneg = 2\recpos$.
Suppose that $m<\recneg/\recpos<M$ for $m>0,M<\infty$. Then according to \citet{PemantleW08}, each critical point contributes some asymptotic factor to the possible upper-bound of the recursion. We denote the matching contribution factor of each point~$\point{x_i}{y_i}$ by $\functwo{T_i}{\recneg}{\recpos}$.

The contribution by each point is given by a different formula, depending on point's multiplicity.

Let~${\mathcal{H}=\partderv{H}{xx}\partderv{H}{yy} - \partderv{H}{xy}}$. Then the contribution $T_i$ of a \emph{multiple point}~$\point{x_i}{y_i}$ is given by:
\begin{equation} \label{eq:mult-points}
    \functwo{T_i}{\recneg}{\recpos} = x_i^{-\recneg} y_i^{-\recpos} \frac{\functwo{G}{x_i}{y_i}}{\sqrt{-x_i^2y_i^2\functwo{\mathcal{H}}{x_i}{y_i}}}.
\end{equation}

Let:
\begin{equation}\label{eq:q-equation}
    \begin{split}
        \functwo{Q}{x}{y} = &-x\partderv{H}{x}y^2\partderv{H}{y}^2 - y\partderv{H}{y}x^2\partderv{H}{x}^2 - y^2\partderv{H}{y}^2x^2\partderv{H}{xx} \\&- x^2\partderv{H}{x}^2y^2\partderv{H}{yy} + 2x\partderv{H}{x}y\partderv{H}{y}xy\partderv{H}{xy},
    \end{split}
\end{equation}
Then the contribution $T_i$ of a \emph{single point}~$\point{x_i}{y_i}$ is given by:
\begin{equation} \label{eq:simple-points}
    \functwo{T_3}{\recneg}{\recpos} = \frac{\functwo{G}{x_3}{y_3}}{\sqrt{2\pi}} {x_3^{-\recneg}} {y_3^{-\recpos}} \sqrt{\frac{-y_3 \functwo{\partderv{H}{y}}{x_3}{y_3}}{\recpos \functwo{Q}{x_3}{y_3}}}.
\end{equation}
In case that $H$ doesn't contain quadratic factors, a point is considered ``single'' if the value of the gradient of $H$ at the point is not zero. Otherwise, it is considered ``multiple''.

\subsubsection{Multiple Points' Contribution}
The points $\point{x_1}{y_1}$ and~$\point{x_2}{y_2}$ from Eq.~\ref{eq:appx-critical-points} are both multiple points.
The corresponding values of $\mathcal{H}$ and $G$ functions for those points are:
\begin{equation*} \label{eq:cal-h-results}
    \begin{split}
        &\functwo{\mathcal{H}}{x_1}{y_1} = \frac{15\sqrt{5}}{2}-\frac{35}{2},
        \functwo{G}{x_1}{y_1} = \sqrt{5}-1, \\
        &\functwo{\mathcal{H}}{x_2}{y_2} = -1,
        \functwo{G}{x_2}{y_2} = 1.
    \end{split}
\end{equation*}
The result for the points' contribution, given by substituting the aforementioned values in Eq.~\ref{eq:mult-points} is:
\begin{equation} \label{eq:mult-points-res}
    \begin{split}
        &\functwo{T_1}{\recneg}{\recpos} \sim \left( \frac{4}{3\sqrt{5}-5} \right) \cdot\left( \frac{1+\sqrt{5}}{2} \right)^{r}  \\
        &\functwo{T_2}{\recneg}{\recpos} = \frac{\functwo{G}{1}{1}}{\sqrt{-\mathcal{H}}} = 1.
    \end{split}
\end{equation}

\subsubsection{Single Points' Contribution}
The point $\point{x_3}{y_3}$ from Eq.~\ref{eq:critical-points} is a single point.
According to Eq.~\ref{eq:q-equation} and Eq.~\ref{eq:simple-points} its contribution is:
\begin{equation}\label{eq:simple-point-res}
    \begin{split}
        \functwo{T_3}{\recneg}{\recpos} &= 
        \frac{(\recneg-\recpos)^{\recneg-\recpos}}{(\recneg-2\recpos)^{\recneg-2\recpos}\cdot \recpos^{\recpos}} 
        \cdot 
        \frac{2\recpos}{\recneg-2\recpos} \cdot 
        \sqrt{\frac{\alpha}{2\pi}},
    \end{split}
\end{equation}
where $\alpha = \bigo(\frac{\recneg^2}{\recpos})$.
%

\subsubsection{Obtained Recurrence Approximation}
%
%

In Sec.~\ref{sec:bound-with-gf}, we presented the motivation for computing a bound for the specific case where there is a linear dependency between $\recneg$ and~$\recpos$, namely, $\recneg = n \recpos$ for a given $n\in \mathbb{N}$. 
For this case, by simply substituting $\recneg$ and reducing the fractions, we get the following factors by each critical point:
\begin{equation} \label{eq:appx-cbs-recursion-est-linear}
    \begin{split}
        &\functwo{T_1}{n\recpos}{\recpos} = 1, \\&
        \functwo{T_2}{n\recpos}{\recpos} = \left( \frac{4}{3\sqrt{5}-5} \right) \cdot \left( \frac{1+\sqrt{5}}{2} \right)^{n\recpos}, \\& 
        \functwo{T_3}{n\recpos}{\recpos} = \left(\frac{(n-1)^{n-1}}{(n-2)^{n-2}} \right)^\recpos 
        \cdot 
        \frac{2}{n-2}
        \cdot
        \sqrt{\frac{\beta}{2\pi \recpos}},
    \end{split}
\end{equation}
where $\beta = \bigo{\left( n^2 \right)}$.

We care for finding the component in Eq.~\ref{eq:appx-cbs-recursion-est-linear} which closely approximates the value of $\functwo{T}{n\recpos}{\recpos}$.
First, observe that~${\frac{2}{n-2} \cdot \sqrt{\frac{\beta}{2\pi}}}$ is asymptotically $\bigo(1)$.
Now, let $n_0$ be the solution for
\begin{equation*}
    \frac{(n-1)^{n-1}}{(n-2)^{n-2}} = \left( \frac{1+\sqrt{5}}{2} \right)^{n},
\end{equation*}
which is $n_0 = \frac{\sqrt{5} + 2}{2} \approx 3.618033$.
Note that we intentionally omit the multiplication by $1/\sqrt{\recpos}$, which is insignificant asymptotically in this case.

Each of the critical points' contribution is used to approximate the recursion value for a certain range of values of $n$.
Empirically, we get that if ${n < n_0}$ then the contribution of the multiple points gives a tight approximation:
\begin{equation*}
    \functwo{T}{n \recpos}{\recpos} \sim 1 + \left( \frac{4}{3\sqrt{5}-5} \right) \cdot \left( \frac{1+\sqrt{5}}{2} \right)^{n\recpos},
\end{equation*}
which is the same, asymptotically, as the $ns$'th Fibonacci number.

If $n \geq n_0$ then the approximation is given by the contribution of the single point, therefore:
\begin{equation*}
    \functwo{T}{n \recpos}{\recpos} \sim \left( \frac{(n-1)^{n-1}}{(n-2)^{n-2}} \right)^\recpos \cdot \frac{1}{\sqrt{\recpos}},
\end{equation*}
where we assume that $\sim$ ignores constant factor within the approximated expression.

We are interested in the asymptotic behavior of this approximation, therefore we focus on the formula obtained for~${n \geq n_0}$.
We notice that $\frac{(n-1)^{n-1}}{(n-2)^{n-2}} < e n$, thus, in conclusion we get the following upper-bound on $T$:
\begin{equation} \label{eq:final-result}
    \functwo{T}{n \recpos}{\recpos} \sim  \frac{\left( e n \right)^\recpos}{\sqrt{\recpos}}.
\end{equation}

\end{appendices}

\bibliography{references.bib}
\end{document}